\documentclass[a4paper,twocolumn,11pt]{quantumarticle}

\pdfoutput=1

\usepackage[utf8]{inputenc}
\usepackage[T1]{fontenc}
\usepackage[english]{babel}

\usepackage{physics}
\usepackage{amsmath}
\usepackage{amssymb}
\usepackage{amsthm}
\usepackage{mathtools}
\usepackage{bm}
\usepackage{bbold}
\usepackage{dsfont}
\usepackage{siunitx}

\usepackage{graphicx}
\usepackage{booktabs}
\usepackage{array}
\usepackage{subfig}

\usepackage[dvipsnames]{xcolor}
\usepackage{soul}
\usepackage[final]{changes}

\usepackage{appendix}    
\usepackage{enumerate}   
\usepackage{hyperref}
\usepackage[numbers]{natbib}

\newtheorem{theorem}{Theorem}

\newtheorem{lemma}[theorem]{Lemma}

\newcommand{\Qsep}{\mathcal{Q}_\omega^{\mathrm{sep}}}
\newcommand{\Qent}{\mathcal{Q}_\omega^{\mathrm{ent}}}
\newcommand{\IcorrMaxSep}{I_\mathrm{corr}^{\omega,\mathrm{sep}}}

\begin{document}

\title{Entanglement in the energy-constrained prepare-and-measure scenario: applications to randomness certification and channel discrimination}

\author{Raffaele D'Avino}
\affiliation{ICFO-Institut de Ci\`encies Fot\`oniques, The Barcelona Institute of Science and Technology,
Av. Carl Friedrich Gauss 3, 08860 Castelldefels (Barcelona), Spain}
\orcid{0009-0009-4592-4075}

\author{Gabriel Senno}
\affiliation{Quside Technologies S.L., C/Esteve Terradas 1, 08860 Castelldefels,
Barcelona, Spain}
\orcid{0000-0002-4788-1664}

\author{Mir Alimuddin}
\affiliation{ICFO-Institut de Ci\`encies Fot\`oniques, The Barcelona Institute of Science and Technology,
Av. Carl Friedrich Gauss 3, 08860 Castelldefels (Barcelona), Spain}
\orcid{0000-0002-5243-085X}

\author{Antonio Acín}
\affiliation{ICFO-Institut de Ci\`encies Fot\`oniques, The Barcelona Institute of Science and Technology,
Av. Carl Friedrich Gauss 3, 08860 Castelldefels (Barcelona), Spain}
\affiliation{ICREA - Instituci\'o Catalana de Recerca i Estudis Avan\c{c}ats, 08010 Barcelona, Spain}
\orcid{0000-0002-1355-3435}

\maketitle

\begin{abstract}
Quantum information tasks are often analyzed under varying trust assumptions about the devices involved. The semi-device-independent (SDI) framework offers a balance between needed assumptions and experimental feasibility. In this work, we study the energy-constrained SDI scenario, where the only assumption in a prepare-and-measure setup is an upper bound on the energy of the prepared quantum states. In contrast to previous studies that restricted the preparation and measurement devices to be classically correlated, we show that allowing entanglement strictly enlarges the set of achievable correlations. We identify two operational consequences of this result. The first concerns randomness certification, where we show that allowing the adversary to employ entangled strategies may significantly reduce the amount of certifiable randomness. This includes situations where the amount of randomness drops to zero in the presence of entanglement, while it remains positive when entanglement is excluded. Second, for the task of distinguishing an arbitrary quantum channel from the identity, we show that the known dimension-independent bound on the advantage conferred by entanglement is violated under an energy constraint.
\end{abstract}

\section{Introduction}
    Quantum information processing tasks are often analyzed under different assumptions about the level of trust in the devices used. Two well-established paradigms in this regard are the \emph{device-dependent} and the \emph{device-independent} (DI) frameworks. In the device-dependent setting, one assumes tomographically complete knowledge of the internal workings of the devices \cite{NielsenChuang2010}, including precise characterizations of the quantum states and measurements involved. While this allows for powerful and efficient protocols, the validity of the results heavily relies on an accurate device calibration, which may not always be feasible in practice. This is especially important in cryptographic applications, where any mismatch between theoretical model and implementation represents a potential security loophole. On the other hand, the device-independent approach removes the need for trusting the internal functioning of the devices, relying instead on the violation of Bell inequalities to certify quantum behavior \cite{Bell1964,CHSH1969,Brunner2014RMP}. This framework provides strong security and robustness guarantees \cite{acin2006bell,Barrett2005,Acin2007DIQKD}, but its implementations are often challenging, requiring entangled states, low noise and high detection efficiencies~\cite{Pironio2009DIQKD,Gisin2010Heralded,Brunner2007AsymDetEff,Lim2013LocalBell,Gerhardt2011FakedBell}. 
    
    The demanding experimental requirements behind the implementation of DI protocols motivates the exploration of intermediate models, such as the \emph{semi-device-independent} (SDI) framework, which offers a middle ground: certified security under bounded assumptions—such as finite dimension—or with one-sided trust via steering, among other variants, and is backed by extensive theoretical and experimental work~\cite{Gallego2010DimWitness,PawlowskiBrunner2011SDIQKD,Branciard2012OneSidedDI,Ahrens2012ExpDimWitness,pauwels2025information,Bowles'2014}. Unlike fully device-independent methods, it can be implemented in a \emph{prepare-and-measure} (PM) configuration, and unlike fully device-dependent schemes, it does not require fine-tuned state preparation or measurement calibration. In this work, we focus on the energy-constrained SDI setting, originally introduced in \cite{van2017semi}, where the only assumption is a bound on the average energy of the quantum states involved: a quantity that is physically meaningful and relatively easy to measure in many experimental platforms. This SDI framework has been experimentally applied for the task of randomness generation~\cite{rusca2019self,rusca2020fast,blazquez2025modulator} with the theory developed in~\cite{van2019correlations}. Our main goal is to understand if entanglement provides any advantage in this SDI scenario.

    Most of the results concerning entanglement advantage in PM setups have been obtained under a bounded dimension. A canonical example is superdense coding, where a \(d\)-dimensional quantum message assisted by a maximally entangled \(d\times d\) state reproduces the correlations attainable with a \(d^{2}\)-level classical system~\cite{bennett1992communication}. In this setting of bounded-dimensional communication (either classical or quantum), the power of pre-shared entanglement has extensively been examined \cite{Pawlowski'2010,Cubitt'2010, Leung2012,Frenkel2022,Renner2023,Alimuddin'2023,Bennett'1999,Bennett2002,tavakoli2021correlations,Piveteau'2022}. Moreover, in a recent work \cite{tavakoli2021correlations}, the set of all correlations achievable in such SDI scenario was characterized by a hierarchy of semidefinite programming relaxations. However, note that the dimension assumption is largely a convenient theoretical abstraction. For instance, in photonic communication setups, the prepared bosonic  states are intrinsically infinite-dimensional. Consequently, any assumption that the transmitted information is confined to a limited number of degrees of freedom can never be strictly enforced and only approximated through a detailed characterization of the devices involved.
    
    Here, we study whether entanglement provides any advantage in PM setups where dimension is unbounded. Under no additional assumptions, no advantage can be observed, as any entangled state can be prepared by the source and shared with the receiver during the protocol execution. That is, the correlations obtained in a PM setup using \(d\)-dimensional quantum preparations assisted by a \(d'\!\times\! d''\) entangled state can always be reproduced by \(n\)-dimensional quantum preparations whenever \(n \ge d\,d''\). We show that this equivalence can be amply violated under the presence of an energy constraint: \emph{qubit} communication even assisted by $2\times 2$ entanglement achieves correlations that no amount of unentangled quantum communication can reproduce. We then provide two operational consequences of the entanglement advantage in the energy-bounded PM scenario. 
    
    First, we focus on the task of randomness certification. Here, the goal is to lower bound the amount of private randomness in the outcomes of the measurement device given some observed nonclassical behavior. We show that the amount of certifiable randomness for a fixed input to the preparation box and given some observed violation significantly decreases if the no-entanglement assumption from \cite{van2017semi,van2019correlations} is lifted, even going to zero from some value of the energy bound. In fact, for almost all values of the energy bound's range, there are entangled PM devices whose outcomes for a fixed input are \emph{deterministic} while certified to be intrinsically random in the absence of entanglement~\cite{van2017semi}.
    
    Second, we look at binary quantum channel discrimination. While entanglement can increase the success probability for this task, it is known that if one of the channels is the identity map, the entanglement advantage is bounded by a constant independent of the systems dimensions. We show that if the inputs to the channels are restricted to those that produce outputs with constrained energy, this bound is violated already for qubit channels. In other words, the entanglement advantage is larger under energy constraints.
        
\section{Preliminaries}
    \subsection{Prepare-and-measure scenario}
        In this work, we focus on one of the simplest quantum information processing scenarios: the PM setup. In this scenario, given an input $x \in \mathcal{X}$, a quantum state $\rho_S^x \in \mathcal{D}(\mathcal{H}_S)$ is prepared. Then, conditioned on another input $y \in \mathcal{Y}$, a measurement is performed, described by a positive operator-valued measure (POVM) $\{M^{b|y}_{S}\}_{b \in \mathcal{B}}$, i.e. by operators $M^{b|y}_{S}\in\mathcal{L}(\mathcal{H}_S)$ satisfying $M^{b|y}_{S} \geq 0$ and $\sum_b M^{b|y}_{S} = \mathbb{1}_{S}$ $\forall y$. Given a state $\rho_{S} \in \mathcal{D}(\mathcal{H}_S)$, the probability of obtaining outcome $b$ is given by $\Tr[M^{b|y}_{S} \rho_{S}]$. Here, $\mathcal{D}(\mathcal{H}_S)$ denotes the set of density matrices over the Hilbert space $\mathcal{H}_S$ and $\mathcal{L}(\mathcal{H}_S)$ the set of linear operators over $\mathcal{H}_S$.
        
        Alternatively, the state-preparation step can be modeled as the application of a quantum channel $\Lambda_x: \mathcal{D}(\mathcal{H}_P) \rightarrow \mathcal{D}(\mathcal{H}_S)$ to a fixed initial state $\sigma_P \in \mathcal{D}(\mathcal{H}_P)$, such that $\rho_S^x = \Lambda_x(\sigma_P)$. We assume that the sets $\mathcal{X}$, $\mathcal{Y}$, and $\mathcal{B}$ are finite. The overall process gives rise to outcome statistics of the form 
        \begin{equation}
            p(b |x,y) = \Tr\left[ M^{b|y}_{S˘} \Lambda_x(\sigma_P) \right].
        \end{equation}
        This scenario is illustrated in Figure~\ref{fig:prep_and_meas}.
        \begin{figure}[t]
            \centering
            \includegraphics[width=0.75\linewidth]{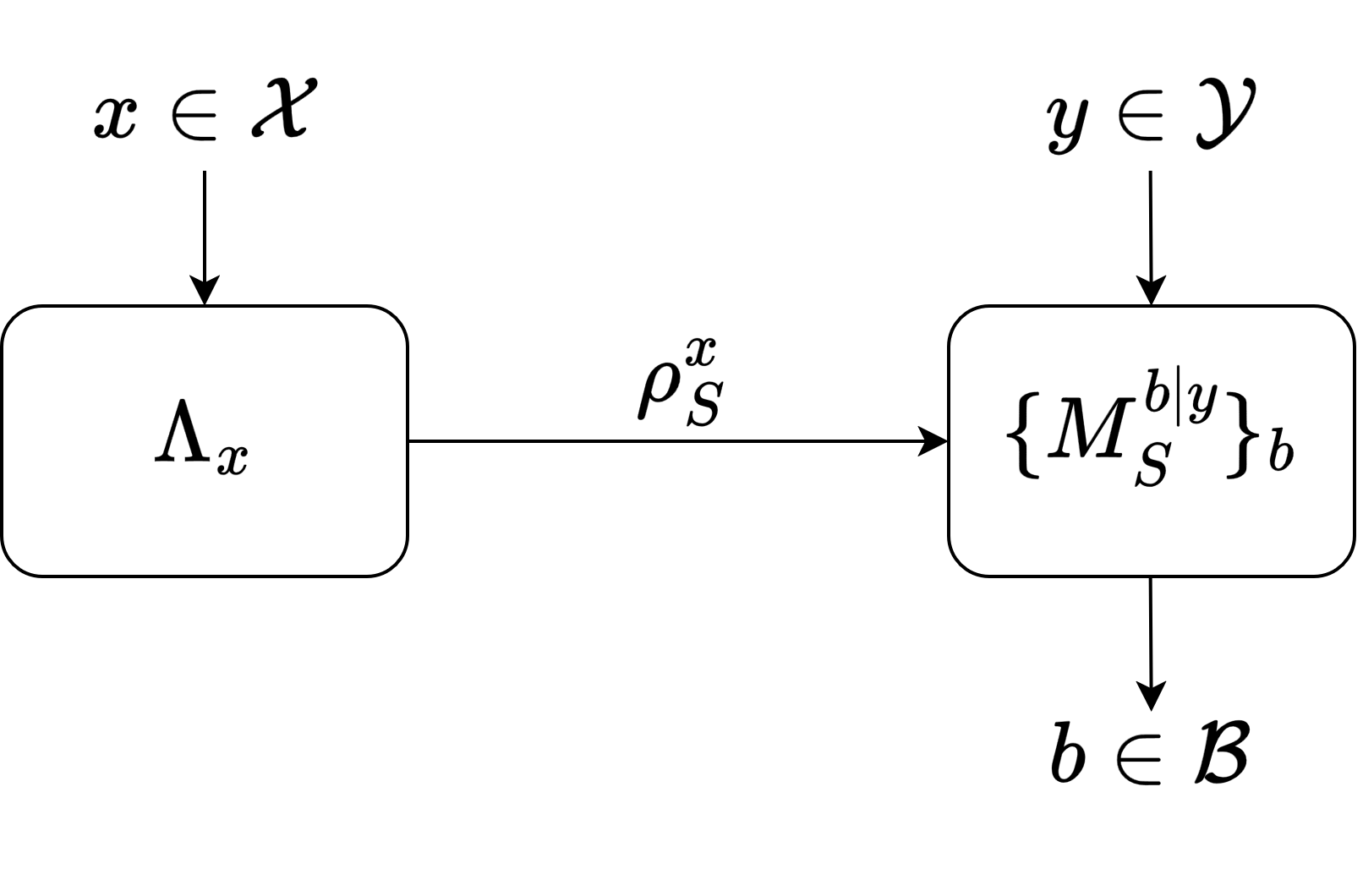}
            \caption{\textbf{PM scenario.} Given an input $x$, a quantum channel $\Lambda_x$ is applied to prepare a state $\rho_S^x$. Subsequently, upon receiving another input $y$, a measurement $\{M^{b|y}_{S}\}_{b}$ is performed on the state, yielding outcome $b$.}
            \label{fig:prep_and_meas}
        \end{figure}
        
        This construction assumes that the preparation and measurement devices are totally independent. More generally, we can consider correlations between them. 
        That is, we may describe the measurement device with a quantum system $M$ such that the composite system $PM$ is in a, in general, correlated quantum state $\sigma_{PM}\in\mathcal{D}(\mathcal{H}_{PM})$. In particular, $\sigma_{PM}$ can be entangled. The measurement is then described by a POVM $\{\Pi_{SM}^{b|y}\}_{b \in \mathcal{B}} \subset \mathcal{L}(\mathcal{H}_{SM})$. The resulting outcome statistics then take the form
        \begin{equation}
            p(b|x,y) = \Tr\left[ \Pi_{SM}^{b|y} (\Lambda_x \otimes \mathbb{1}_M)(\sigma_{PM}) \right].
        \end{equation}
        This scenario is depicted in figure \ref{fig:EA prep_and_meas}.
        \begin{figure}[t]
            \centering
            \includegraphics[width=0.75\linewidth]{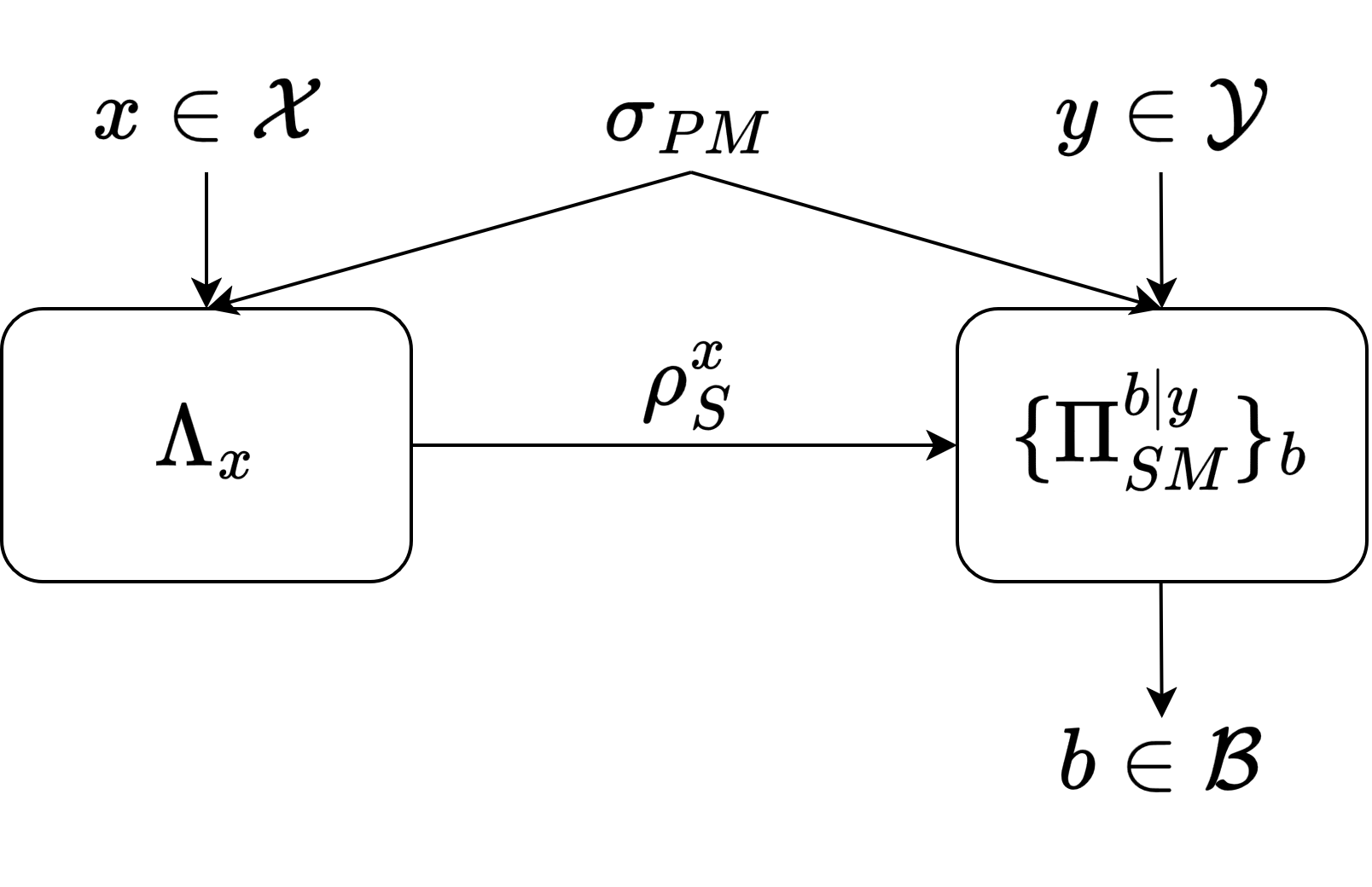}
            \caption{\textbf{Entanglement-assisted PM scenario.} Starting from a shared state $\sigma_{PM}$, channel $\Lambda_x$ acts on $\Tr_M[\sigma_{PM}]$ upon input $x$ to produce $\rho_S^x$, which is then measured with $\{\Pi_{SM}^{b|y}\}_{b}$ upon input $y$, yielding outcome $b$.}
            \label{fig:EA prep_and_meas}
        \end{figure}
        
        In this work, we restrict our analysis to the minimal scenario involving only two inputs $x\in\{0,1\}$ and a single binary measurement $\{\Pi_{SM}^{b}\}_{b\in\{0,1\}}$, the observed correlations being $p(b|x)$, with no measurement input $y$.
    \subsection{Energy-constrained setting}\label{sec:energy-constrained-scenario}
        In the SDI framework introduced in \cite{van2017semi}, the channels $\{\Lambda_x\}_x$ and the measurement $\{\Pi_{SM}^b\}_b$ are uncharacterized. The main assumption is a constraint on the \emph{energy} of the prepared quantum states. More formally, it is assumed that 
        \begin{equation}
            \exists~H~ \text{s.t.} \Tr\left[ H\, \Lambda_x(\sigma_P) \right] \leq \omega_x,\label{eq:energy-constraint}
        \end{equation}
        where $\omega_x\in[0,1/2)$ and $H$ is some hermitian operator (i.e., the Hamiltonian) with a non-degenerate ground state and a unit gap, i.e. $H=\sum_{i\geq 0}E_i\ket{E_i}\bra{E_i}$ with $E_0=0$ and $E_i\geq1$ for all $i\geq1.$ This constraint is physically motivated, as measuring the energy of a quantum system is typically feasible in experimental settings (see, e.g., \cite{van2017semi} for a discussion of how to implement this type of constraints). Notice that no bounds are imposed on the dimension of the underlying Hilbert spaces, which may, in principle, be infinite-dimensional. For simplicity, we will henceforth assume that $\omega_0=\omega_1=\omega$ \footnote{This applies, in particular, to the experiments in \cite{rusca2020fast,blazquez2025modulator}}. By the restrictions imposed on $H$, we can equivalently write the energy constraint as a requirement on the prepared states of having a large overlap with a common ground state $\ket{g}$, that is
        \begin{equation}
            \exists~\ket{g}\in\mathcal{H}_S~\text{s.t.}~\bra{g} \Lambda_x(\sigma_P)\ket{g} \geq 1-\omega.\label{eq:energy-constraint}
        \end{equation}
        We will alternate between these two formulations henceforth. An additional assumption in \cite{van2017semi} is that the PM boxes can only be \emph{classically-correlated}. Formally, this entails restricting the initial shared state $\sigma_{PM}$ to be separable. Given an energy bound $\omega$, we will denote with $\Qsep$ the set of \emph{behaviors} $\mathbf{p}=\{p(b|x)\}_{b,x}$ achievable under this restriction and $\Qent$ when entanglement is allowed. In the next section, we show that
        \begin{equation}\label{eq:separation}
            \Qsep \subsetneq \Qent. 
        \end{equation}
        
\section{Entanglement enlarges the set of energy-constrained correlations}\label{sec:separation}
    
    Let $E_x\in[-1,+1]$ be the \emph{correlators}
    $E_x=p(0|x)-p(1| x)$ and let \begin{equation}\label{eq:icorr}
            I_{\text{corr}}(\mathbf{p}) := |E_0 - E_1|.
        \end{equation}
    In \cite{van2017semi}, it is shown that
        \begin{equation}\label{maxIclasscorr}
        \begin{aligned}
            \forall \mathbf{p}\in\Qsep~&I_{\text{corr}}(\mathbf{p})\leq 4\sqrt{\omega(1 - \omega)}=:I^{\omega,\mathrm{sep}}_{\text{corr}}.
        \end{aligned}
    \end{equation}
    We will prove Eq. \eqref{eq:separation} by finding behaviors $\mathbf{p}\in\Qent$ violating this inequality.

    We performed a seesaw optimization to find lower bounds to the maximum value of Eq. \eqref{eq:icorr} over behaviors in $\Qent$. To address this, we recasted the problem into an equivalent formulation (see Appendix~\ref{App:formulation equivalence}) that is more amenable to analysis. 
The seesaw approach alternates between optimizing over a subset of the positive operators of the problem, while keeping the others fixed, thereby reducing the problem to a sequence of semidefinite programs (SDPs). Each step solves a convex subproblem, and the process is repeated until convergence to a (typically local) optimum.\\

The results obtained are plotted in Fig. \ref{fig:EA_analytic}. As can be seen in the figure,
$$
\max_{\mathbf{p}\in\Qent}I_\mathrm{corr}(\textbf{p})>\IcorrMaxSep,
$$
implying Eq. \eqref{eq:separation}. Within our numerics, the violation is already achieved with a qubit communication system (\(\dim\mathcal{H}_S=2\)) when the preparation and measurement devices share a \(2\times 2\) entangled state (i.e., \(\dim\mathcal{H}_P=\dim\mathcal{H}_M=2\)). Allowing a \(3\times 3\) shared entangled state while keeping the communicated system a qubit yields a strictly larger violation. Notably, we do not observe any higher violations when further increasing any of the systems' dimensions. We thus conjecture that qubit communication assisted by \(3\times 3\) entanglement is optimal. The advantage is illustrated in Fig.~\ref{fig:EA_analytic}. For the \(2\times 2\) case, we also provide an explicit analytical construction that attains the reported value (see Appendix~\ref{App:analytical state}).
\begin{figure}[t]
        \centering
    \includegraphics[width=1\linewidth]{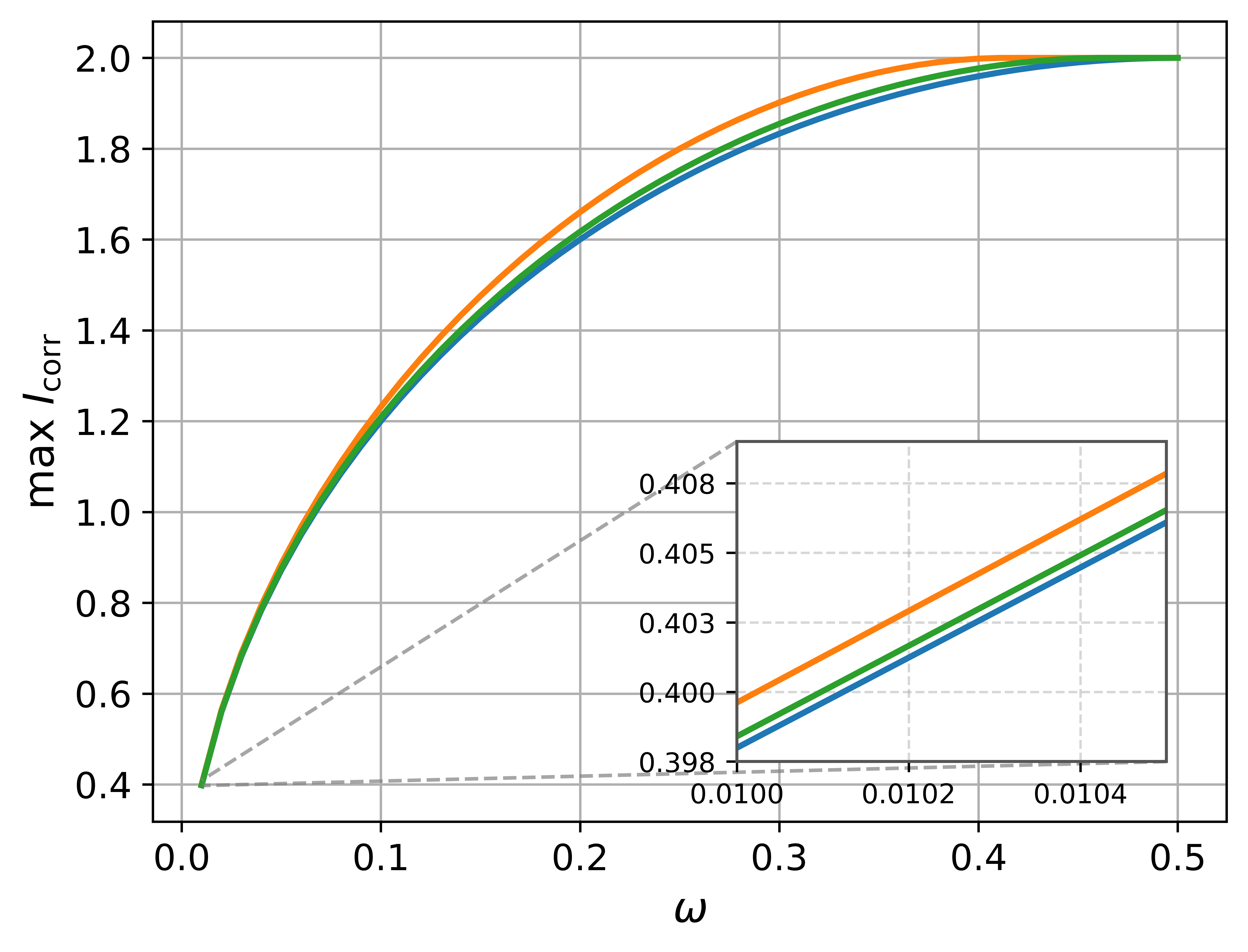}
        \caption{\textbf{Seesaw violations of Eq. \eqref{maxIclasscorr}.}
        In {\color{blue}blue}, the classically correlated bound \(I^{\omega,\mathrm{sep}}_{\mathrm{corr}}\) from Eq.~\eqref{maxIclasscorr} is shown. Curves in \textcolor{ForestGreen}{green} and \textcolor{orange}{orange} display violations of this bound when a qubit communication (\(\dim \mathcal{H}_S=2\)) is assisted by \(2\times 2\) and \(3\times 3\) shared entanglement between the PM devices, respectively.}
        \label{fig:EA_analytic}
    \end{figure}

       For the behaviors we obtained numerically, we have observed that in order to obtain values of Eq. \eqref{eq:icorr} larger than $I^{\omega,\mathrm{sep}}_\mathrm{corr}$ at least one of the states $\{(\Lambda_x\otimes I_M)(\sigma_{PM})\}_x$ must be mixed. Notice that, since we can assume the state $\sigma_{PM}$ to be pure without loss of generality \footnote{Let $\ket{\psi}_{PMM'}$ be a purification of $\sigma_{PM}$, i.e. $\Tr_{M'}[\ket{\psi}\bra{\psi}]=\sigma_{PM}$. Then, for every set of channels $\{\Lambda_x\}_x$, ground state $\ket{g}$ and POVM $\{\Pi_{SM}^b\}_b$, if we let $\rho_{SMM'}^x=\Lambda_x\otimes\mathbb{1}_{MM'}(\ket{\psi}\bra{\psi})$, we have that
       $\Tr[\ket{g}\bra{g}\otimes \mathbb{1}_{MM'}\rho_{SMM'}^x]=\Tr[\ket{g}\bra{g}\otimes \mathbb{1}_{M}(\Lambda_x\otimes\mathbb{1}_{M})\sigma]$
       and $\Tr[\Pi_{SM}^b\otimes \mathbb{1}_{M'}\rho_{SMM'}^x]=\Tr[\Pi_{SM}^b(\Lambda_x\otimes\mathbb{1}_{M})(\sigma)]$.}, 
       that would imply that at least one of the channels $\{\Lambda_x\}_x$ must be nonunitary. Theorem \ref{th:unitary channels} below states that this is indeed necessary. 
    
    \begin{theorem}\label{th:unitary channels}
        Let $\{U_x\}_x$ be unitary operators on $\mathcal{H}_P=\mathcal{H}_S$, $\ket{\psi}_{SM}\in\mathcal{H}_{SM}$ and $\ket{g}_S\in\mathcal{H}_S$ such that 
        \begin{equation*}
        \bra{\psi_x}(\ket{g}\bra{g}\otimes \mathbb{1}_{M})\ket{\psi_x} \geq 1-\omega,
        \end{equation*}
        for $\ket{\psi_x}=(U_x\otimes \mathbb{1}_{M})\ket{\psi}$ and some $\omega\in[0,1/2)$. Then, for every {\rm POVM} $\{\Pi_{SM}^b\}_b$, it holds that
        \begin{equation*}
            \{\bra{\psi_x}\Pi_{SM}^b\ket{\psi_x}\}_{b,x}\in\Qsep.
        \end{equation*}
    \end{theorem}
    The proof of Thm. \ref{th:unitary channels} can be found in Appendix \ref{App:unitary channels}. 
    The fact that entanglement enlarges the set of energy-constrained behaviors has operational consequences. This is the focus of the following section.
    
\section{Operational consequences}
In this section, we show how, under energy constraints, the presence of entanglement impacts the tasks of SDI randomness certification and binary channel discrimination (see \cite{mondal2025classical} for recent, related results on the classical capacity of quantum channels).

    \subsection{Decrease in the amount of certifiable randomness}
    In the task of randomness certification, the goal is to be able to conclude that the measurement outcomes contain some degree of unpredictability (or, intrinsic randomness) solely from the observed statistics and the assumptions on the devices.

    Given $\omega\in[0,1/2)$ and an expected value $I^\mathrm{exp}_\mathrm{corr}$ for Eq. \eqref{eq:icorr}, we will focus on the amount of certifiable randomness for a fixed input $x=x^*$ to the preparation box. This is the quantity of interest in randomness generation protocols of the \textit{spot-checking} type. In such protocols, a small fraction of the rounds are devoted to test the nonclassicality of the device by choosing the inputs at random, while the remaining rounds are left for randomness generation using a predefined input. Without loss of generality, we fix $x^*=0$.
    
    To quantify the amount of intrinsic randomness, we consider an eavesdropper, Eve, who holds a quantum system $E$ potentially correlated with the honest user's PM device. Eve's goal is to \emph{guess} the device's outcome. That is, she performs a measurement $\{M_E^e\}_e$ on her system trying to maximize the correlation with the device's outcome. The probability that Eve guesses correctly is given by
    \begin{equation}\label{eq:Pguess}
        \footnotesize
        \begin{aligned}
        &P_{\text{guess}}(B|E,X=0,I^\mathrm{exp}_\mathrm{corr},\omega)
        :=&&\\ 
        &\operatorname*{max} \limits_{\substack{\{\Lambda_x\}_x,\{\Pi_{SM}^{b}\}_b\\
        \ket{g}, \ket{\psi}, \{M_E^e\}_e}}\sum_e \Tr\Big[
        (\Pi_{SM}^{e}\otimes M_E^e)
        (\Lambda_{0}\otimes\mathbb{1}_{ME})
        (\ket{\psi}\bra{\psi})
        \Big]&&\\
        &\quad~\textrm{subject to } &\\
        &\Tr[(\ket{g}\bra{g}\otimes\mathbb{1}_{ME})
        (\Lambda_x\otimes\mathbb{1}_{ME})
        (\ket{\psi}\bra{\psi})] \ge 1-\omega,&&\\
        &\sum_{b,x}(-1)^{b\oplus x}\Tr\Big[
        (\Pi_{SM}^b\otimes\mathbb{1}_E)
        (\Lambda_x\otimes\mathbb{1}_{ME})
        (\ket{\psi}\bra{\psi})
        \Big] \ge I_\text{corr}^\text{exp}&&
        \end{aligned}
    \end{equation}
where, without loss of generality, we have assumed that the global state $\ket{\psi}_{PME}$ is pure. We then quantify the amount of intrinsic randomness in the device's outcomes via the conditional min-entropy
    \begin{align}
        H_{\min}^*=-\log_2 P_{\text{guess}}(B|E,X=0,I^\mathrm{exp}_\mathrm{corr},\omega)\label{eq:hmin}.
    \end{align}
    
    Let us denote with $H_{\min}^{*,\mathrm{sep}}$ the value of Eq. \eqref{eq:hmin} that results from restricting the reduced state of systems $P$ and $M$ in Eq. \eqref{eq:Pguess} to be separable (i.e. when the devices are only classically correlated). $H_{\min}^{*,\mathrm{sep}}$ can be lower-bounded using the techniques developed in \cite{van2019correlations}. On the other hand, when the initial shared state is arbitrary, we do not have, at the moment, any computational procedure to lower bound $H_\mathrm{min}^*$ (let alone compute). Nevertheless, we can obtain upper bounds by applying a seesaw strategy using a reformulation of Eq. \eqref{eq:Pguess} more amenable for it (see Appendix \ref{App:quantum_guess_prob}). If the obtained upper bounds, for given $\omega\in[0,1/2)$ and $I^\mathrm{exp}_\mathrm{corr}$, are below the corresponding lower bounds for the separable case, this amounts to showing that the assumption of only classicaly-correlated devices in \cite{van2017semi,van2019correlations} indeed leads to an overestimation of the amount of intrinsic randomness generated.
    
    \begin{figure}[t]
        \centering
        \includegraphics[width=1\linewidth]{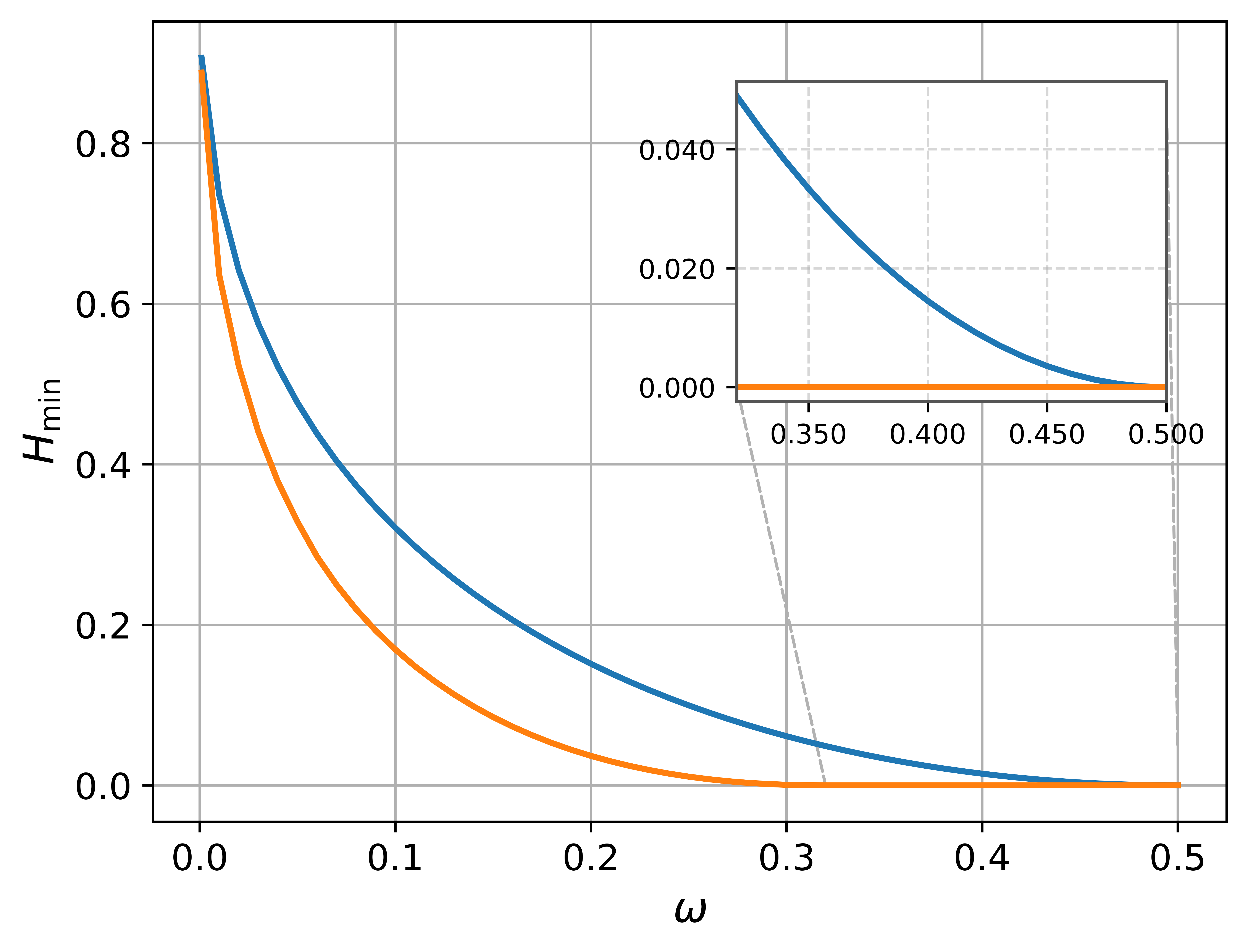}
        \caption{\textbf{$H_\mathrm{min}$ lower bound given a violation $\IcorrMaxSep$}. In {\color{blue}blue}, the certifiable randomness $H^{*,\mathrm{sep}}_\mathrm{min}$ under a no-entanglement assumption coming from \cite{van2019correlations}. In {\color{orange}orange}, a seesaw upper bound on the certifiable randomness $H^{*}_\mathrm{min}$ in the entanglement-assisted scenario.}
        \label{fig:guess_prob}
    \end{figure}

     In Fig. \ref{fig:guess_prob} we plot $H^*_\mathrm{min}$ and $H^{*,\mathrm{sep}}_\mathrm{min}$ for the maximum classicaly-correlated value $\IcorrMaxSep$ of Eq. \eqref{eq:icorr} as a function of $\omega$ \footnote{Notice that, in order to be able to compare $H^*_\mathrm{min}$ and $H^{*,\mathrm{sep}}_\mathrm{min}$ we must choose $I^\mathrm{exp}_\mathrm{corr}\leq \IcorrMaxSep$, because otherwise the optimization for $H^{*,\mathrm{sep}}_\mathrm{min}$ would be infeasible.}. We see that, indeed, the amount of certifiable randomness significantly decreases when the only-classically-correlated-devices assumption is lifted. Remarkably, $H^{*}_\mathrm{min}$ even drops to zero (up to numerical precision) for $\omega\gtrsim 0.32$, while $H^{*,\mathrm{sep}}_\mathrm{min}$ remains nonzero for the whole $\omega$ range. In other words, there are entangled PM devices achieving a value $\IcorrMaxSep$ of Eq. \eqref{eq:icorr} whose outcomes are \emph{completely determined} by Eve's side-information. In fact, something even stronger holds, as we show next.

     In \cite[Sec. 6.1]{van2017semi}, the set $\overline{\mathcal{D}}_{x^*}^\omega$ of correlators $(E_0,E_1)$ which arise from convex combinations of behaviors $\mathbf{p}^\lambda \in\Qsep$ that are \emph{determinstic} on input $x=x^*$, i.e. with $E^\lambda_{x^*}\in\{-1,+1\}$, was shown to be defined by the equations
    \begin{align}
        (E_{x^*}+1)(1-\frac{4\omega}{1+E_{x^*}})^2-E_{\overline{x^*}}&\leq 1\label{eq:determ1},\\
        (E_{x^*}-1)(1-\frac{4\omega}{1-E_{x^*}})^2-E_{\overline{x^*}}&\geq -1,\label{eq:determ2}
    \end{align}
    where $\overline{x^*}=1-x^*.$ As stated in \cite{van2017semi}, if a classically correlated PM device exhibits correlators $(E_0,E_1)$ outside $\overline{\mathcal{D}}_{x^*}^\omega$, this certifies that there is some randomness in the outcomes for input $x^*$, as they cannot be convex decompositions of behaviors deterministic for that input. Let us see now that this conclusion no longer holds when entanglement is allowed.

    Notice that if we restrict to behaviors $\mathbf{p}$ with $p(0|0)=1$ (and, hence, $E_0=1$), Eq. \eqref{eq:determ1} then implies
    \begin{equation}
        E_1\geq 2(1-2\omega)^2-1=:I_\mathrm{det}^{\omega,\mathrm{sep}}\label{eq:deterministic-ineq-for-fixed-input-0}
    \end{equation}
    for all $\mathbf{p}\in\Qsep$.
    In Fig. \ref{fig:deterministic-ineq} we compare $I_\mathrm{det}^{\omega,\mathrm{sep}}$ to seesaw upper bounds to 
    \begin{equation}
        \min_{\mathbf{p}\in\Qent:p(0|0)=1} E_1.\label{eq:minimization-of-E1} 
    \end{equation}
    \begin{figure}[t]
        \centering
        \includegraphics[width=1\linewidth]{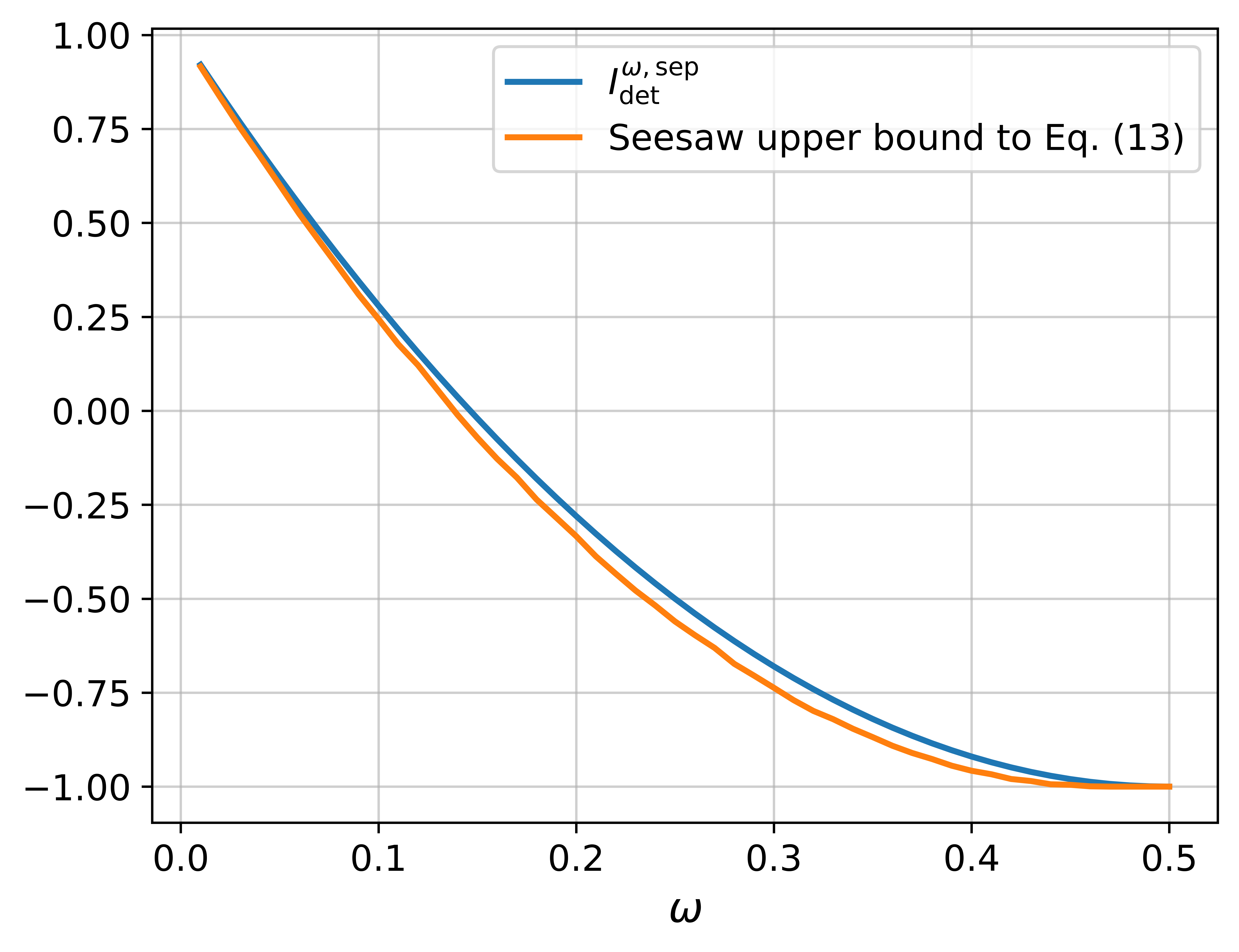}
        \caption{\textbf{Entangled violations of Eq.~\eqref{eq:deterministic-ineq-for-fixed-input-0}}. In {\color{blue}blue}, the lower bound $I^{\omega,\mathrm{sep}}_\mathrm{det}$ in Eq. \eqref{eq:deterministic-ineq-for-fixed-input-0}. In {\color{orange}orange}, seesaw upper bounds to Eq. \eqref{eq:minimization-of-E1}. We observe that entangled PM devices can generate behaviors deterministic for one of the inputs and, yet, violate the inequality in Eq. \eqref{eq:deterministic-ineq-for-fixed-input-0}.}
        \label{fig:deterministic-ineq}
    \end{figure}
    
    
    Numerically, we find that for all tested positive values of $\omega$, there exist behaviors $\mathbf{p} \in \Qent$ that are deterministic for input $x = 0$ but violate Eq.~\eqref{eq:deterministic-ineq-for-fixed-input-0}. We conjecture that this holds for all $\omega > 0$. In other words, for any positive $\omega$, the assumption of only classically correlated devices is necessary to conclude that a violation of Eq.~\eqref{eq:deterministic-ineq-for-fixed-input-0} certifies intrinsic randomness for input $x = 0$.

    \subsection{Channel discrimination}
    We now turn to the task of \emph{binary quantum channel discrimination}. Here, given two channels, $\Lambda_0$ and $\Lambda_1$, the goal is to determine which has acted on a quantum system. Ancillary systems and entangled inputs may be used to optimize the success probability \cite{Sacchi,Piani,Matthews}. The standard procedure is as follows: an input state $\sigma_P$ (possibly entangled with an ancilla) is prepared, one of the channels is applied, and a measurement is performed to infer which channel acted.  
    
    When the two channels are equally likely, the optimal success probability is
    \begin{equation}
        P_{\text{g}}^\diamond = \frac{1}{2} + \frac{1}{4} \| \Lambda_0 - \Lambda_1 \|_\diamond,
    \end{equation}
    where $\|\cdot\|_\diamond$ is the \emph{diamond norm}, defined by
    \begin{equation}
        \| \Lambda \|_\diamond := \sup_{\sigma_{PM}\in\mathcal{D}(\mathcal{H}_{PM})} \| (\Lambda \otimes \mathbb{1}_M)(\sigma_{PM}) \|_1,
    \end{equation}
    with the supremum over all density operators $\sigma_{PM}$ on system plus ancilla, and $\|\cdot\|_1$ the trace norm. On the other hand, when entanglement with an ancilla is not allowed/considered, the optimal success probability is given by
    \begin{equation}
        P_{\text{g}}^1 = \frac{1}{2} + \frac{1}{4} \| \Lambda_0 - \Lambda_1 \|_1,
    \end{equation}
    where $\|\cdot\|_1$ is the \emph{induced trace norm}, defined by
    \begin{equation}\label{induced_trace_norm}
        \begin{aligned}
            \|\Lambda\|_1 &= \sup_{\sigma_P\in\mathcal{D}(\mathcal{H}_P)}\|\Lambda(\sigma_P)\|_1 \\
            &= \max_{\sigma_P\in\mathcal{D}(\mathcal{H}_P),\,\|M\|_\infty\leq 1} \Tr[M(\Lambda(\sigma_P))],
        \end{aligned}
        \end{equation}
        where $\|\cdot\|_{\infty}$ denotes the infinity norm and the second line follows from Schatten norm duality. 
        
    Entanglement can enhance discrimination, as there are cases where no separable strategy achieves the success probability attainable with entangled inputs. For instance, Werner-Holevo channels can have arbitrarily small induced trace norm while maintaining a diamond norm of $1$  (see \cite[Example 3.36]{watrous2018theory}). However, this discrepancy does not occur when distinguishing a channel from the identity. In this case, the following theorem holds ~\cite{watrous2018theory} (see also \cite{kretschmann2004tema}):  
    
   \begin{theorem}{\cite[Thm. 3.56]{watrous2018theory}}\label{Th:id_discrim}
    Let $\Lambda$ be a quantum channel such that $\|\Lambda - \mathbb{1}\|_1 \in [1,2]$. Then
    \begin{equation}
        \frac{\|\Lambda - \mathbb{1}\|_\diamond}{\|\Lambda - \mathbb{1}\|_1} \le \sqrt{2}.
    \end{equation}
    \end{theorem}
    
    This theorem limits the advantage entanglement can provide in distinguishing a channel $\Lambda$ from the identity. Specifically, if the channel is ``sufficiently far'' from the identity, the ratio between the optimal success probability with entanglement and without is bounded by
    \begin{equation}\label{entangl_adv}
    \begin{aligned}
        P_{\text{adv}}(\Lambda) := \frac{P_{\text{g}}^{\diamond}}{P_g^{1}} &= \frac{\frac{1}{2}(1 + \frac{1}{2} | \Lambda - \mathbb{1} |_\diamond)}{\frac{1}{2}(1 + \frac{1}{2} | \Lambda - \mathbb{1} |_1)}\\
        &\leq \left(\frac{1}{2} + \frac{1}{\sqrt{2}}\right) \approx 1.208.
    \end{aligned}
    \end{equation}
    
    Operationally, this task corresponds to detecting whether a particular transformation $\Lambda$ has been applied or if the system remains unchanged. This instance of channel discrimination is of particular interest, as it directly relates to the fundamental capacity of identifying the action of a quantum process. 
    
    \subsubsection{Energy-Constrained Channel Discrimination}
In this section we impose an \emph{energy constraint} on channel discrimination tasks: for any admissible entangled probe through the channels, the resulting output must have energy at most \(\omega\), evaluated with respect to a fixed Hamiltonian \(H\) that satisfies the conditions in Section~\ref{sec:energy-constrained-scenario}. This restriction can be motivated by power consumption requirements or by input power range of the detector in the measurement device.

Consider the pair \(\{\mathbb{1},\,\Lambda\}\). Because one element is the identity, the output-energy bound immediately enforces the same bound on the input state. To quantify any gain from preshared entanglement in this setting, we use
\(P_\mathrm{adv}^\mathrm{ec}(\Lambda)\), namely the advantage defined in Eq.~\eqref{entangl_adv} evaluated under the energy-constrained (ec) conditions.
        We will show that there exists a family $\{\Lambda_\omega\}_\omega$ of quantum channels such that:
        \begin{enumerate}
            \item For all $\omega\in(0,1/2]$, $P_\mathrm{adv}^\mathrm{ec}(\Lambda_\omega)>P_\mathrm{adv}(\Lambda_\omega)$;  and
            \item For $\omega\gtrsim 0.39$, $P_\mathrm{adv}^\mathrm{ec}(\Lambda_\omega)$ violates the upper bound in Eq. \eqref{entangl_adv}. 
        \end{enumerate}
        
        As mentioned in Section \ref{sec:separation}, in Appendix ~\ref{App:analytical state} we provide an analytical family of states $\{\ket{\psi^{0,\omega}}_{SM},\rho_{SM}^{1,\omega}\}_\omega$ violating Eq. \eqref{maxIclasscorr} (for an appropriate measurement) for every $\omega$. We take $\ket{\psi^{0,\omega}}_{SM}$ as the initial shared state and we construct the channel $\Lambda_\omega$ using state-channel duality with the state $\rho_{SM}^{1,\omega}$ (see Appendix~\ref{App:analytical state}).

         Our aim is to quantify the entanglement-assisted advantage in this energy-constrained scenario. To this end, we compute the ratio between the diamond norm and the induced trace norm of the difference between the two channels, both in the unconstrained and energy-constrained cases, thereby evaluating the quantity in Eq.~\eqref{entangl_adv}.

        \medskip
        
        \textbf{Upper bound on the entanglement advantage without energy constraint.}
        To compute the diamond distance between two channels, we express the problem as a semidefinite program (SDP), which can be efficiently solved for low-dimensional systems. Using the Choi operator of a channel, this can be written as~\cite{skrzypczyk2023semidefinite}  
        \begin{equation}
            \|\Lambda\|_\diamond = \max_{Y,\sigma} \Tr[Y X_\Lambda],
        \end{equation}
        subject to
        \begin{equation}
        \begin{aligned}
            -d(\mathbb{1} \otimes \sigma) \leq Y \leq d(\mathbb{1} \otimes \sigma), \\
            \sigma \geq 0,\quad \Tr[\sigma] = 1,
        \end{aligned}
        \end{equation}
        where \( X_\Lambda \) is the Choi matrix of \( \Lambda \) and \( d \) is the Hilbert-space dimension on which the channel acts.
        
        To evaluate the quantity in Eq.~\eqref{entangl_adv}, we also need the induced trace distance between the channels $\{\Lambda_\omega\}_\omega$ and the identity channel \( \mathbb{1} \). Unlike the diamond norm, the induced trace norm does not admit an SDP formulation, as it involves a non-convex optimization. Nevertheless, lower bounds can be obtained via a seesaw optimization based on Eq. \eqref{induced_trace_norm}. 
        The optimization alternates between the state \( \rho \) and the operator \( M \), providing a lower bound on the induced trace norm.
        
        Finally, this provides an upper bound to Eq.~\eqref{entangl_adv}:
        \begin{equation}\label{eq:padv-upper-bound-not-ec}
            P_{\text{adv}}(\Lambda_\omega) \leq \frac{1 + \frac{1}{2}\|\Lambda_\omega - \mathbb{1}\|_\diamond}{1 + \frac{1}{2}\text{\emph{lower}}\|\Lambda_\omega - \mathbb{1}\|_1},
        \end{equation}
        thus bounding the potential advantage of entanglement in the absence of energy constraints.
        
        \medskip

        \textbf{Lower bound on the entanglement advantage in the energy-constrained scenario.}
        To compare with the energy-constrained setting, we aim to bound the same quantities: specifically, a \emph{lower bound} on the diamond norm and an \emph{upper bound} on the induced trace norm under the energy constraint. Together, these provide a lower bound on the right-hand side of Eq.~\eqref{entangl_adv}.  

        The states $\ket{\psi^{0,\omega}}_{SM}$ and $(\Lambda_\omega\otimes\mathbb{1})(\ket{\psi^{0,\omega}}_{SM})$ satisfy the energy constraint by construction. Therefore, their trace distance is a lower bound to the energy-constrained diamond norm between $\mathbb{1}$ and $\Lambda_\omega$.
        
        The upper bound on the energy-constrained induced trace norm is obtained via the \emph{Lasserre hierarchy}, a method for polynomial optimization that provides a sequence of SDP relaxations converging to the global optimum (see Appendix~\ref{App:Lasserre hierarchy}).
        
        Finally, we estimate a lower bound on the ratio between the guessing probabilities with and without entanglement in the energy-constrained case:
        \begin{equation}\label{eq:padv-lower-bound-ec}
            P_{\text{adv}}^\mathrm{ec}(\Lambda_\omega)
            \geq
            \frac{1+\frac{1}{2}\text{\emph{lower}}\|\Lambda_\omega - \mathbb{1}\|_\diamond^\mathrm{ec}}{1+\frac{1}{2}\text{\emph{upper}}\|\Lambda_\omega - \mathbb{1}\|_1^\mathrm{ec}}.
        \end{equation}
        
        In Figure~\ref{fig:guess_prob_adv} we plot \eqref{eq:padv-upper-bound-not-ec} and \eqref{eq:padv-lower-bound-ec} as a function of $\omega$. As can been seen, the entanglement advantage is always greater in the energy-constrained setting, even violating the universal bound in Eq. \eqref{entangl_adv} for a sufficiently large $\omega$. 
        
        \begin{figure}[t]
            \centering
            \includegraphics[width=1\linewidth]{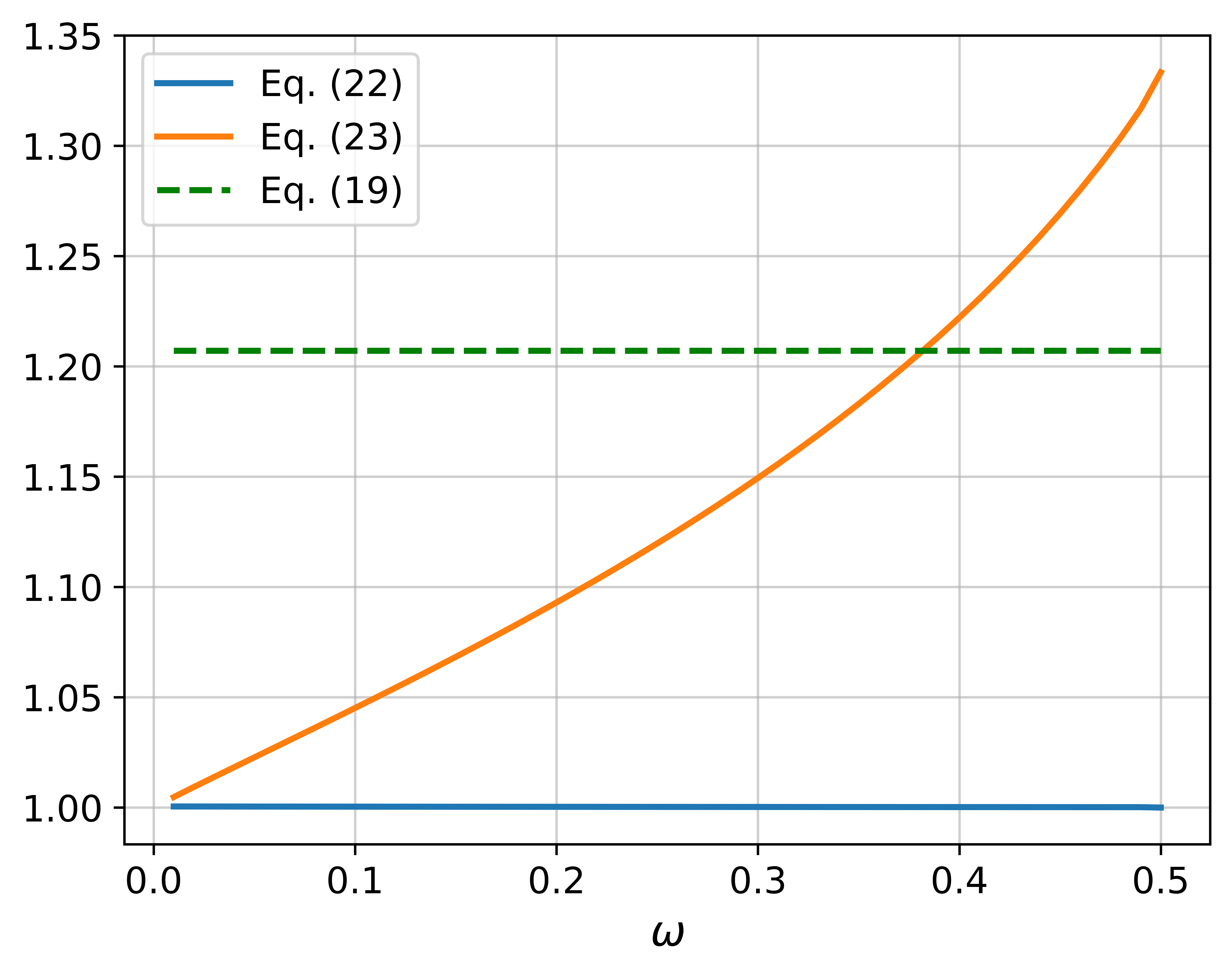}
            \caption{\textbf{Entanglement advantage under energy constraints}. In {\color{blue}blue}, the upper bound to $P_{\text{adv}}(\Lambda_\omega)$ from Eq.~\eqref{eq:padv-upper-bound-not-ec}. In {\color{orange}orange}, the lower bound to $P_{\text{adv}}^{\text{ec}}(\Lambda_\omega)$ from Eq.~\eqref{eq:padv-lower-bound-ec}. In {\color{ForestGreen}green}, the upper bound from  Eq.~\eqref{entangl_adv} to $P_{\text{adv}}(\Lambda)$ for an arbitrary channel $\Lambda$.}
            \label{fig:guess_prob_adv}
        \end{figure}
     
\section{Conclusion}
    We have investigated the role of shared entanglement between preparation and measurement devices in the energy-constrained SDI framework. Our results demonstrate that allowing such correlations enlarges the set of achievable statistics beyond what is possible without entanglement. This has concrete consequences for quantum information tasks. First, in randomness certification, we showed that even when observed statistics admit a non-entangled realization, the presence of shared entanglement can substantially reduce the amount of certifiable randomness, reaching zero for some range of values of the energy bound. Second, in quantum channel discrimination, we constructed an example where one of the two channels is the identity, for which the entanglement-assisted success probability approaches a $4/3$ improvement over the non-assisted case. This advantage not only surpasses the corresponding unconstrained scenario for the same pair of channels, but also exceeds what is achievable in the unconstrained setting for discrimination of the identity against any channel. 
    
    These results underline the importance of explicitly accounting for the role of entanglement in semi-device-independent analyses. The assumption of no shared entanglement has a decisive impact on the performance and security guarantees of quantum protocols. Future work may focus on characterizing the set of correlators achievable in this more general setting and on deriving lower bounds for the amount of certifiable randomness.

\section*{ACKNOWLEDGEMENTS}
This work was supported by the Government of Spain (Severo Ochoa CEX2019-000910-S, FUNQIP and European Union NextGenerationEU PRTR-C17.I1), Fundació Cellex, Fundació Mir-Puig, Generalitat de Catalunya (CERCA program), and the EU Quantera project Veriqtas. 
RD acknowledges funding from the European
Union’s Horizon Europe research and innovation programme under the Marie Skłodowska-Curie grant agreement No. 101081441. GS acknowledges funding from the Government of Spain (Torres Quevedo PTQ2021-011870). M.A. acknowledges the funding supported by the European Union  (Project- QURES, Grant Agreement No. 101153001).

\section*{Note added}
While finishing this manuscript, we became aware of a similar work by Roch i Carceller et al., which reports related results.

\section*{Data and code availability}
The Python code used to perform the numerical optimizations and generate the figures presented in this work is available at \url{https://github.com/raffaeledavino/EntInEnerConstrPM}.

\bibliographystyle{quantum}
\bibliography{References}

\clearpage
\onecolumngrid
\appendix

\begin{center}
\large{\textbf{\textsc{Supplementary Material}}}

\end{center}

\section{Nonunitary channels are neccessary to observe $\Qsep\subsetneq\Qent$}\label{App:unitary channels}
We prove here a result stronger than Thm. \ref{th:unitary channels} from the main text, which will then follow as a corollary.

\begin{theorem}\label{thm:mixed-states}
Let $\ket{\psi_1}_{SM},\ket{\psi_2}_{SM}\in \mathcal{H}_S\otimes \mathcal{H}_M$ be such that 
\begin{equation}
\ket{\psi_2}_{SM} = (U_{S}\otimes \mathbb{1}_M)\ket{\psi_1}_{SM},
\end{equation}
for some unitary $U_{S}$. Let $\ket{g}_{S}\in \mathcal{H}_S$ satisfy
\begin{equation}
\bra{\psi_x}(\ket{g}_{S}\bra{g}\otimes \mathbb{1}_M)\ket{\psi_x} \ge 1 - \omega, 
\qquad x = 1,2.
\end{equation}
Then, for every measurement $\{\Pi^b_{SM}\}_b$,
\begin{equation}
    \left\{ \bra{\psi_x}{\Pi^b_{SM}}\ket{\psi_x} \right\}_{b,x}
    \in \Qsep.
\end{equation}

\end{theorem}

\begin{proof}
We will show that there exists a state $\ket{\tilde{g}}\in \mathcal{H}_S\otimes \mathcal{H}_M$ 
having an overlap $\ge 1-\omega$ with both $\ket{\psi_1}$ and $\ket{\psi_2}$. Hence, any behavior resulting from a measurement $\{\Pi^b_{SM}\}_b$ on these states can be implemented in a classically correlated setup, in which the prepared \emph{single} system is now $SM$ and the groundstate is $\ket{\tilde{g}}$.

\medskip

Let $\braket{\psi_1}{\psi_2} = r e^{i\theta}$ and define
\begin{equation}
\ket{\phi} := \ket{\psi_1} + e^{-i\theta}\ket{\psi_2},
\end{equation}
which is unnormalized. Then,
\begin{equation}
\bigl|\braket{\phi}{\psi_x}\bigr| = 1 + r,
\qquad 
\|\phi\|^2 = 2(1 + r).
\end{equation}
Let $\ket{\psi_\perp}$ be orthogonal to both $\ket{\psi_1}$ and $\ket{\psi_2}$, 
and define
\begin{equation}
\ket{\tilde{g}}_{SM}
:= \frac{\ket{\phi} + \xi \ket{\psi_\perp}}
{\|\ket{\phi} + \xi \ket{\psi_\perp}\|},
\qquad
\xi := \sqrt{\frac{(1+r)^2}{1-\omega} - 2(1+r)}.
\end{equation}
In Lemma~\ref{lemma:overlap}, we will show that $\xi \in [0,\infty)$. 
With this choice, we have
\begin{equation}
|\braket{\tilde{g}}{\psi_x}|^2 = 1 - \omega.
\end{equation}
Therefore, the bipartite system $SM$ can be viewed as a new single system $S'$, 
whose ground state is $\ket{\tilde{g}}$ and such that 
$\ket{\psi_x}_{S'}$ has overlap $\ge 1 - \omega$ with it. 
Thus, irrespective of the measurement $\{\Pi_{S'}^b\}_b$, 
the resulting behavior always belongs to 
$\Qsep$.
\end{proof}

\begin{lemma}\label{lemma:overlap}
Let $\ket{\psi_1}_{SM},\ket{\psi_2}_{SM}\in \mathcal{H}_S\otimes \mathcal{H}_M$ satisfy 
\begin{equation}
\ket{\psi_2}_{SM} = (U_{S}\otimes \mathbb{1}_M)\ket{\psi_1}_{SM},
\end{equation}
for a unitary $U_{S}$. Suppose there exists $\ket{g}_{S}\in \mathcal{H}_S$ such that 
\begin{equation}
p \le p_x := \bra{\psi_x}(\ket{g}_{S}\bra{g}\otimes \mathbb{1}_M)\ket{\psi_x},
\qquad x = 1,2.
\end{equation}
Then,
\begin{equation}
|\braket{\psi_1}{\psi_2}| \ge 2p - 1.
\end{equation}
\end{lemma}

\begin{proof}
Let $\rho_S := \mathrm{Tr}_M(\ket{\psi_1}\!\bra{\psi_1})$, 
and define $\ket{h} := U^\dagger\ket{g}$. 
Let $\Pi$ be the projector onto the two-dimensional subspace
\begin{equation}
S := \mathrm{span}\{\ket{g}, \ket{h}\} \subseteq \mathcal{H}_S.
\end{equation}
Set $\sigma := \Pi \rho_S \Pi \succeq 0$ and $t := \mathrm{Tr}\,\sigma \in [0,1]$. 
Notice that
\begin{equation}\label{eq:diagsp}
\bra{g}\sigma\ket{g} = p_1, 
\qquad 
\bra{h}\sigma\ket{h} = p_2.
\end{equation}
Let $T := \Pi U \Pi : S \to S$. Then $T$ is a contraction 
($\|T\|_{\mathrm{op}} \le 1$) and 
\begin{equation}
T\ket{h} = \Pi U \ket{h} = \Pi \ket{g} = \ket{g}.
\end{equation}
Since $\Pi\sigma = \sigma$, we have
\begin{equation}
\mathrm{Tr}(\sigma T) = \mathrm{Tr}(\sigma \Pi U \Pi) = \mathrm{Tr}(\sigma U).
\end{equation}
We now show that
\begin{equation}\label{eq:bound_lemma}
|\mathrm{Tr}(\sigma T)| \ge 2p - t.
\end{equation}
Choose an orthonormal basis $\{\ket{g},\ket{g_\perp}\}$ of $S$ and write
\begin{equation}
\ket{h} = \cos\theta\,\ket{g} + e^{i\varphi}\sin\theta\,\ket{g_\perp},
\qquad
\ket{h_\perp} = -e^{-i\varphi}\sin\theta\,\ket{g} + \cos\theta\,\ket{g_\perp},
\end{equation}
with $\theta \in [0,\pi/2]$.  
In this basis, using \eqref{eq:diagsp} and $\mathrm{Tr}\,\sigma=t$, we have that
\begin{equation}
\sigma = 
\begin{pmatrix}
p_1 & z \\
z^* & t - p_1
\end{pmatrix},
\qquad z \in \mathbb{C}.
\end{equation}
Without loss of generality, we assume that $p_1$ is chosen so that $p_2 \ge p_1$. Define
\begin{equation}
a := \bra{g}\sigma\ket{h},
\qquad
b := \bra{g_\perp}\sigma\ket{h_\perp}.
\end{equation}
Expanding $\bra{h}\sigma\ket{h} = p_2$ yields
\begin{equation}
\mathrm{Re}(e^{i\varphi} z) 
\ge \frac{2p - t}{2}\cdot \frac{\sin\theta}{\cos\theta}.
\end{equation}
From this, one obtains
\begin{align}
\mathrm{Re}(a) + \mathrm{Re}(b) &= \cos\theta\, t, \nonumber\\
\mathrm{Re}(a) - \mathrm{Re}(b) &\ge \frac{2p - t}{\cos\theta}, \nonumber\\
\mathrm{Im}(a) &= \mathrm{Im}(b). 
\end{align}
Hence,
\begin{equation}\label{eq:bound_diff_module}
|a|^2 - |b|^2 
= (\mathrm{Re}(a) + \mathrm{Re}(b))(\mathrm{Re}(a) - \mathrm{Re}(b)) 
  + (\mathrm{Im}(a) + \mathrm{Im}(b))(\mathrm{Im}(a) - \mathrm{Im}(b))
\ge t(2p - t).
\end{equation}
Since $T$ is a contraction and $T\ket{h} = \ket{g}$, 
the ``defect'' operator $D := \mathbb{1}_S - T^\dagger T \succeq 0$  obeys $\bra{h}D\ket{h}=0$, which implies $D\ket{h}=0$ and
$T^\dagger T\ket{h}=h$. Thus for every $\ket{y}\perp \ket{h}$,
\begin{equation}
\braket{Ty}{Th} = \bra{y}T^\dagger T\ket{h} = \braket{y}{h} = 0.
\end{equation}
Thus $T\ket{h_\perp}$ is orthogonal to $\ket{g}$, 
and since $T$ is a contraction, $T\ket{h_\perp} = \beta\ket{g_\perp}$ 
for some $|\beta|\le 1$. Hence,
\begin{equation}
T\ket{h} = \ket{g}, \qquad T\ket{h_\perp} = \beta\,\ket{g_\perp}.
\end{equation}
This gives
\begin{equation}
\mathrm{Tr}(\sigma T) = a^* + \beta b^*.
\end{equation}
Therefore,
\begin{equation}
|\mathrm{Tr}(\sigma T)| 
= |a^* + \beta b^*| 
\ge \min_{|\beta|\le 1} |a^* + \beta b^*| 
= \max(|a| - |b|, 0).
\end{equation}
Using the bound from \eqref{eq:bound_diff_module},
\begin{equation}
|a| - |b| 
= \frac{|a|^2 - |b|^2}{|a| + |b|} 
\ge \frac{t(2p - t)}{|a| + |b|}.
\end{equation}
By applying the Cauchy--Schwarz inequality to the $\sigma$-inner product $\langle v, w \rangle_\sigma := \braket{v}{\sigma|w}$ (defined on the support of $\sigma$), we obtain $|a| \le \sqrt{p_1 p_2}$ and $|b| \le \sqrt{(t - p_1)(t - p_2)}$. Consequently, $|a| + |b| \le t$, since $p \le t$. Thus, when $2p-t\ge 0$,
\begin{equation}
|a|-|b|\geq \frac{t(2p-t)}{t}=2p-t.
\end{equation}
If $2p-t\le0$, the right-hand side of \eqref{eq:bound_lemma} is nonpositive and the inequality is trivial. Thus,
\begin{equation}
|\mathrm{Tr}(\sigma T)| \ge 2p - t.
\end{equation}
Finally, decompose $\rho_S = \sigma + \tau$, 
where $\tau := (\mathbb{1}_{S} - \Pi)\rho_S(\mathbb{1}_{S} - \Pi) \succeq 0$ 
and $\mathrm{Tr}\,\tau = 1 - t$. 
By the reverse triangle inequality and Hölder’s inequality,
\begin{equation}
|\braket{\psi_1}{\psi_2}|
= |\mathrm{Tr}(\rho_S U)|
\ge |\mathrm{Tr}(\sigma U)| - |\mathrm{Tr}(\tau U)|
\ge (2p - t) - (1 - t)
= 2p - 1.
\end{equation}
\end{proof}
\bigskip

Thm. \ref{th:unitary channels} from the main text follows form Thm. \ref{thm:mixed-states} by setting $\ket{\psi_1}=(U_1\otimes\mathbb{1}_{M})\ket{\psi}$ and $U=U_2U_1^\dagger$.

\section{Mapping Between the Channel-Based and State-Based Problems}\label{App:formulation equivalence}
    We are concerned with the optimization \begin{equation}
    \begin{aligned}
         \max_{\mathbf{p}\in\Qent} I_\text{corr}(\mathbf{p})&=\max_{\substack{\ket{g}_S,\{\Pi_{SM}^{b}\}_b, \\[1pt] \sigma_{PM}, \{\Lambda_x\}_x}} \sum_{b,x}(-1)^{b\oplus x}\Tr\left[\Pi_{SM}^b(\Lambda_x\otimes\mathbb{1}_M)\left(\sigma_{PM}\right)\right],
    \end{aligned}
    \end{equation}
    subject to the constraint
    \begin{equation}
        \Tr\!\left[(\ket{g}_S\bra{g}\otimes\mathbb{1}_{M})(\Lambda_x \otimes \mathbb{1}_{M})(\sigma_{PM})\right] \geq 1-\omega.
    \end{equation}
    We consider an equivalent formulation given by the state-channel duality \cite{choi1975completely} in which, instead of optimizing explicitly over the local channels $\{\Lambda_x\}_x$, we optimize over the resulting states $\{\rho_{SM}^x\}_x$ on systems $SM$. That is,
    \begin{align}
        \max_{\mathbf{p}\in\Qent} I_\text{corr}(\mathbf{p})&=\max_{\substack{\ket{g}_S,\{\Pi_{SM}^{b}\}_b, \\[1pt] \{\rho_{SM}^x\}_x}} \sum_{b,x}(-1)^{b\oplus x}\Tr\left[\Pi_{SM}^b\rho_{SM}^x\right], 
    \end{align}
    with constraints
    \begin{equation}
    \begin{aligned}
        \Tr\!\left[(\ket{g}_{S}\bra{g} \otimes \mathbb{1}_{M})\rho_{SM}^{x}\right] &\geq 1-\omega, \\
        \Tr_S[\rho_{SM}^{0}] &= \Tr_S[\rho_{SM}^{1}].
    \end{aligned}
    \end{equation}
    \section{Analytical constructions}\label{App:analytical state}
        Based on the numerical evidence, we assume that the state prepared with input $x=0$ is the pure state
    \begin{equation}\label{eq:pure_state}
        \ket{\psi^{0,\omega}}_{SM}=\sqrt{1-\omega}\ket{00}+\sqrt{\omega}\left(\sqrt{p}\ket{10}+\sqrt{1-p}\ket{11}\right).
    \end{equation}
    Imposing this form does not affect the maximization of $I_\text{corr}=|E_0 - E_1|$, where $E_0$ and $E_1$ are the correlators defined as $E_x = p(0|x)-p(1|x)$. From the numerical results, we further observe that the maximum of $I_\text{corr}$ can be attained by constraining $\rho^{1,\omega}_{SM}$ to have the spectral decomposition
    \begin{equation}
        \rho^{1,\omega}_{SM}=q\ket{\phi^\omega}\bra{\phi^\omega}+(1-q)\ket{01}\bra{01},
    \end{equation}
    where $\ket{\phi^\omega}_{SM}$ can be expanded in the computational basis with real coefficients as
    \begin{equation}
        \ket{\phi^\omega}_{SM}=\pm \sqrt{a}\ket{00}\pm \sqrt{b}\ket{10}\pm\sqrt{1-a-b}\ket{11}.
    \end{equation}
    The parameters $a,b,q$ are determined by the constraints on the energy and on the reduced state. Letting $\ket{g}=\ket{0}$, the energy condition reads
    \begin{equation}\label{eq:energy}
        \Tr[(\ket{0}\bra{0}\otimes\mathbb{1}_{M})\rho^{1,\omega}_{SM}]=1-\omega 
        \iff qa+1-q=1-\omega . 
    \end{equation}
   The marginal condition requires
    \begin{equation}\label{eq:marginal_state}
    \begin{aligned}
        &\Tr_{S}[\rho^{1,\omega}_{SM}] = \Tr_{S}[\ket{\psi^{0,\omega}}_{SM}\bra{\psi^{0,\omega}}] \iff \\
        &\quad q(a+b)\ket{0}\bra{0} + \big[q(1-a-b) + 1 - q\big] \ket{1}\bra{1} \\
        &\quad \pm q\sqrt{b(1-a-b)} (\ket{0}\bra{1} + \ket{1}\bra{0}) \\
        &= (1-\omega+\omega p)\ket{0}\bra{0} + \omega(1-p)\ket{1}\bra{1} \\
        &\quad + \omega\sqrt{p(1-p)} (\ket{0}\bra{1} + \ket{1}\bra{0}).
    \end{aligned}
    \end{equation}
    In Eq.~\eqref{eq:marginal_state}, the sign $\pm$ depends on whether the coefficients of $\ket{10}$ and $\ket{11}$ in $\ket{\phi^\omega}_{SM}$ have the same or opposite sign. This condition can be expressed as
    \begin{subequations}\label{eq:marginal}
    \begin{align}
        q(a+b)&=1-\omega +\omega p,\label{eq:marginal_a}\\
        q(1-a-b)+a-q&=\omega(1-p),\\
        \pm q\sqrt{b(1-a-b)}&=\omega\sqrt{p(1-p)}. \label{eq:marginal_c}
    \end{align}
    \end{subequations}
    Note that Eq.~\eqref{eq:marginal_c} cannot be satisfied in the case of the minus sign. Additionally, it is natural to expect that the sign of the coefficient of $\ket{00}$ differs from the signs of the other two coefficients, in order to minimize the fidelity with $\ket{\psi^{0,\omega}}_{SM}$. We can therefore assume that the state $\ket{\phi^{\omega}}_{SM}$ is
    \begin{equation}
        \ket{\phi^{\omega}}_{SM}=- \sqrt{a}\ket{00} + \sqrt{b}\ket{10}+\sqrt{1-a-b}\ket{11}.
    \end{equation}
    The four constraints obtained are not independent and we can consider the set given by Eqs.~\eqref{eq:energy},~\eqref{eq:marginal_a},~\eqref{eq:marginal_c}. Solving this set, we obtain two solutions. The first is $(a,b,q)=(1-\omega,p\omega,1)$ that doesn't give any advantage and the second is
    \begin{equation}\label{eq:parameters}
    \begin{aligned}
        a&=(2p\omega-2\omega+1)/(2p\omega-\omega+1),\\
        b&=\omega(1-p)/(2p\omega-\omega+1),\\
        q&=2p\omega-\omega+1.
    \end{aligned}
    \end{equation}
    Note that, in this way, we are not constraining $a,b$ and $q$ to be probabilities. $0\leq a\leq1$, $0\leq b\leq1$ and $q\geq0$ are satisfied $\forall p\in[0,1]$ and $\forall \omega\in[0,1/2]$, but $q\leq1\implies p\leq 1/2$. Expliciting the terms in Eq.~\eqref{eq:parameters}, we can rewrite the state $\rho^{1,\omega}_{SM}$ as
    \begin{equation}
    \begin{aligned}
        \rho^{1,\omega}_{SM}=&(2p\omega-2\omega+1)\ket{00}\bra{00}+(\omega-2p\omega)\ket{01}\bra{01}\\
        &+\omega(1-p)\ket{10}\bra{10}+\omega p\ket{11}\bra{11}\\
        &-\sqrt{(2p\omega-2\omega+1)\omega(1-p)}\left[\ket{00}\bra{10}+\ket{10}\bra{00}\right]\\
        &-\sqrt{(2p\omega-2\omega+1)\omega p}\left[\ket{00}\bra{11}+\ket{11}\bra{00}\right]\\
        &+\omega\sqrt{p(1-p)}\left[\ket{10}\bra{11}+\ket{11}\bra{10}\right].
    \end{aligned}
    \end{equation}

    This construction of the prepared states can be recast into an analytical framework involving an initial entangled state and two quantum channels. Specifically, we assume that the state $\ket{\psi^{0,\omega}}_{SM}$ defined in Eq.~\eqref{eq:pure_state} is the initial shared state. In this setting, one of the two channels can be taken as the identity, i.e., $\Lambda_{0,\omega} = \mathbb{1}$. To determine the other channel, we proceed as follows
    \begin{equation} \label{eq:channel_kraus}
    \begin{aligned}
        \rho_{SM}^{1,\omega} &= (\Lambda_{1,\omega} \otimes \mathbb{1}_{M})(\sigma_{PM}) \\
                     &= (\Lambda_{1,\omega} \otimes \mathbb{1}_{M})(\ket{\psi^{0,\omega}}_{SM}\!\bra{\psi^{0,\omega}}) \\
                     &= \sum_i (E_i^{1,\omega} \otimes \mathbb{1}_{M}) \ket{\psi^{0,\omega}}_{SM}\!\bra{\psi^{0,\omega}} ((E_i^{1,\omega})^\dagger \otimes \mathbb{1}_{M}) \\
                     &= \sum_i (E_i^{1,\omega} A_{\omega} \otimes \mathbb{1}_{M}) \ket{\phi^+}_{SM}\!\bra{\phi^+} (A_\omega^\dagger (E_i^{1,\omega})^\dagger \otimes \mathbb{1}_{M}),
    \end{aligned}
    \end{equation}
    where the last line uses the Kraus representation of $\Lambda_{1,\omega}$ and the fact that any bipartite pure state can be expressed as a local transformation of the maximally entangled state $\ket{\phi^+}_{SM}$.
    
    Consider now the spectral decomposition of $\rho_{SM}^{1,\omega}$
    \begin{equation}
        \rho_{SM}^{1,\omega} = \sum_i \lambda_i^{1,\omega} \ket{\eta_i^{1,\omega}}\!\bra{\eta_i^{1,\omega}}.
    \end{equation}
    Each eigenvector $\ket{\eta_i^{1,\omega}}$ can itself be realized by applying a local channel $\tilde{\Lambda}_i^{1,\omega}$ on either subsystem $S$ or $M$ of $\ket{\phi^+}_{SM}$
    \begin{equation} \label{eq:channel_spectrum}
        \rho_{SM}^{1,\omega} = \sum_i \lambda_i^{1,\omega} \, (\tilde{\Lambda}_i^{1,\omega} \otimes \mathbb{1}_{M}) \ket{\phi^+}_{SM}\!\bra{\phi^+} ((\tilde{\Lambda}_i^{1,\omega})^\dagger \otimes \mathbb{1}_{M}).
    \end{equation}
    By comparing Eqs.~\eqref{eq:channel_kraus} and \eqref{eq:channel_spectrum}, it follows that each Kraus operator can be expressed as
    \begin{equation}
        E_i^{1,\omega} = \sqrt{\lambda_i^{1,\omega}}\, \tilde{\Lambda}_i^{1,\omega} A^{-1}_{\omega}.
    \end{equation}
    To explicitly identify the operator $A_{\omega}$, we write the Schmidt decomposition of $\ket{\psi^{0,\omega}}_{SM}$
    \begin{equation}
        \ket{\psi^{0,\omega}}_{SM} = \sqrt{\alpha_{\omega}}\, \ket{\eta_{\omega}} \otimes \ket{0} + \sqrt{\beta_{\omega}}\, \ket{\eta_{\omega}^\perp} \otimes \ket{1}.
    \end{equation}
    From this decomposition, the operator $A_{\omega}$ is given by
    \begin{equation}
        A_{\omega} = \sqrt{2\alpha_{\omega}}\, \ket{\eta_{\omega}}\!\bra{0} + \sqrt{2\beta_{\omega}}\, \ket{\eta_{\omega}^\perp}\!\bra{1}.
    \end{equation}
    Indeed, applying $A_{\omega} \otimes \mathbb{1}_{M}$ to the maximally entangled state $\ket{\phi^+}_{SM}$ reproduces $\ket{\psi^{0,\omega}}_{SM}$
    \begin{equation}
    \begin{aligned}
        (A_{\omega} \otimes \mathbb{1}_{M}) \ket{\phi^+}_{SM} 
        &= \left(\left(\sqrt{2\alpha_{\omega}}\, \ket{\eta_{\omega}}\!\bra{0} + \sqrt{2\beta_{\omega}}\, \ket{\eta_{\omega}^\perp}\!\bra{1}\right) \otimes \mathbb{1}_{M}\right) \cdot \frac{1}{\sqrt{2}} (\ket{00} + \ket{11}) \\
        &= \sqrt{\alpha_{\omega}}\, \ket{\eta_{\omega}} \otimes \ket{0} + \sqrt{\beta_{\omega}}\, \ket{\eta_{\omega}^\perp} \otimes \ket{1} \\
        &= \ket{\psi^{0,\omega}}_{SM}.
    \end{aligned}
    \end{equation}
    In a similar way, the channels \(\{\tilde{\Lambda}_i^{1,\omega}\}_i\) can be associated to each eigenvector \(\ket{\eta_i^{1,\omega}}\), thereby completing the reconstruction of the Kraus $\{E_i^{1,\omega}\}_i$.

    The parameter $p$ used in the analytical construction of the states and channels is determined as a function of $\omega$ by optimizing $I_{\text{corr}}$ in Eq.~\eqref{eq:icorr} for Fig.~\ref{fig:EA_analytic}, and the upper bound to $P_{\text{adv}}^{ec}(\Lambda_{\omega})$ in Eq.~\eqref{eq:padv-lower-bound-ec} for Fig.~\ref{fig:guess_prob_adv}.

\section{Quantum guessing probability}\label{App:quantum_guess_prob}
The most general attach that an adversary could do is to hold a purification of the state and perform a measurement on his part. The \emph{quantum guessing probability} is then defined as
\begin{equation}
    P_{\text{guess}}^{Q}(B\mid E,X=x^*,I_\text{corr}^\text{exp},\{\omega_x\}_x)=\max_{\substack{\ket{g}_S,\{\Pi_{SM}^{b}\}_{b},\; \{M_{E}^{e}\}_e,\\
    \{\Lambda_x\}_x,\ket{\psi}_{PME}}}\sum_e\Tr[(\Pi_{SM}^{e}\otimes M_{E}^{e})(\Lambda_{x^*}\otimes\mathbb{1}_{ME})(\ket{\psi}_{PME}\bra{\psi})],
\end{equation}
subject to
\begin{equation}
    \begin{aligned}
        &\Tr[(\ket{g}_{S}\bra{g}\otimes\mathbb{1}_{ME})(\Lambda_x\otimes\mathbb{1}_{ME})\left(\ket{\psi}_{PME}\bra{\psi}\right)]\geq 1-\omega_x,\\
        &\sum_{b,x}(-1)^{b\oplus x}\Tr\left[\Pi_{SM}^b(\Lambda_x\otimes\mathbb{1}_{ME})(\ket{\psi}_{PME}\bra{\psi})\right]\geq I_\text{corr}^\text{exp}.
    \end{aligned}
\end{equation}
This problem can be reformulated by introducing the post-measurement states 
\begin{equation}
    \sigma_{PM}^e=\Tr_E[(\mathbb{1}_{PM}\otimes M_E^e)(\ket{\psi}_{PME}\bra{\psi})]/p(e),
\end{equation}
which can be interpreted as the state held by the honest party after Eve has performed her measurement and obtained outcome $e$. By substituting this new quantity, we obtain
\begin{equation}
    P_{\text{guess}}^{Q}(B\mid E,X=x^*,I_\text{corr}^\text{exp},\{\omega_x\}_x)=\max_{\ket{g}_S,\{\Pi_{SM}^{b}\}_{b},\{p(e),\sigma_{PM}^{e}\}_{e}}\sum_{e}p(e)\Tr[\Pi_{SM}^{e}(\Lambda_{x^*}\otimes\mathbb{1}_M)(\sigma_{PM}^{e})],
\end{equation}
subject to
\begin{equation}
    \begin{aligned}
        &\Tr[(\ket{g}_{S}\bra{g}\otimes\mathbb{1}_M)(\Lambda_x\otimes\mathbb{1}_{M})\left(\sum_e p(e)\sigma_{PM}^e\right)]\geq 1-\omega_x,\\
        &\sum_{b,x}(-1)^{b\oplus x}\Tr\left[\Pi_{SM}^b(\Lambda_x\otimes\mathbb{1}_M)\left(\sum_e p(e)\sigma_{PM}^e\right)\right]\geq I_\text{corr}^\text{exp}.
    \end{aligned}
\end{equation}

\section{Lasserre Hierarchy}\label{App:Lasserre hierarchy}
Let 
\(\mathbb{R}[x] = \mathbb{R}[x_1, \dots, x_n]\) 
denote the space of real polynomials in \(n\) variables. Consider an arbitrary polynomial
\begin{equation}
    f(x) = \sum_{\alpha \in \mathbb{N}^n} f_\alpha \, x^\alpha \in \mathbb{R}[x].
\end{equation}
Any sequence of real numbers 
\(\mathbf{y} = \{y_\alpha\}_{\alpha \in \mathbb{N}^n}\), indexed by multi-indices corresponding to the monomials \(x^\alpha\), defines a linear functional 
\(L_\mathbf{y} : \mathbb{R}[x] \to \mathbb{R}\) via
\begin{equation}
    L_\mathbf{y}(f) := \sum_{\alpha \in \mathbb{N}^n} f_\alpha \, y_\alpha.
\end{equation}
In this way, one can identify the dual space of \(\mathbb{R}[x]\) with the space of all real sequences \(\mathbf{y} = \{y_\alpha\}_\alpha\). If $\mathbf{y}$ arises from a probability measure $\mu$ supported on a set $K \subseteq \mathbb{R}^n$, namely
\begin{equation}
    y_\alpha = \int_K x^\alpha \, d\mu(x),
\end{equation}
then $L_\mathbf{y}(f)$ coincides with the integral of $f$ against $\mu$. This duality motivates the \emph{Lasserre hierarchy}, a sequence of semidefinite relaxations for polynomial optimization problems of the form
\begin{equation}
    p^\star = \inf_{x \in K} f(x),
\end{equation}
where $K$ is defined by polynomial inequalities $g_j(x) \geq 0$. Clearly, if $\mu$ is supported on $K$, one has $p^\star \le L_\mathbf{y}(f)$, so optimizing over feasible moment sequences yields a relaxation of the original problem. For a given relaxation order $d$, one restricts to truncated sequences $\mathbf{y} = \{y_\alpha\}_{|\alpha|\le 2d}$. Feasibility of $\mathbf{y}$ is enforced via semidefinite constraints. The \emph{moment matrix} $M_d(\mathbf{y})$ with entries
\begin{equation}
    M_d(\mathbf{y})(\alpha,\beta) = L_\mathbf{y}(x^{\alpha+\beta}), \quad |\alpha|,|\beta|\le d,
\end{equation}
must be positive semidefinite to guarantee $L_\mathbf{y}(p(x)^2)\ge 0$ for all polynomials $p$ of degree at most $d$. Similarly, for each polynomial inequality $g_j(x)\ge 0$, the \emph{localizing matrix}
\begin{equation}
M_{d-d_j}(g_j \mathbf{y})(\alpha,\beta) = L_\mathbf{y}\!\big(g_j(x)\, x^{\alpha+\beta}\big), 
\qquad d_j = \lceil \deg(g_j)/2 \rceil,
\end{equation}
is required to be positive semidefinite to enforce $L_\mathbf{y}(g_j(x)p(x)^2)\ge 0$ \cite{laurent2008sums}. The resulting level-$d$ relaxation is the semidefinite program
\begin{equation}
    \begin{aligned}
    \min_{y}\quad & \sum_{|\alpha|\le 2d} f_\alpha \,y_\alpha,\\
    \text{s.t.}\quad 
    & M_d(\mathbf{y})\;\succeq\;0,\\
    & M_{d-d_j}(g_j\,\mathbf{y})\;\succeq\;0,\quad j=1,\dots,m,\\
    & y_0 = 1.
    \end{aligned}
\end{equation}
Under mild conditions, such as compactness of $K$ (typically ensured via the Archimedean property), the sequence of SDP optimal values $\{p_d^\star\}_d$ is nondecreasing and converges to the global optimum $p^\star$ as $d\to\infty$ \cite{laurent2008sums}.\\
To apply this framework to the induced trace norm in Eq.~\ref{induced_trace_norm}, we parametrize the independent entries of the state $\rho$ and operator $M$ as
\begin{equation}
    \rho = \begin{pmatrix}
    \rho_{00} & \rho_{01}^{R}+i\rho_{01}^{I} \\
    \rho_{10}^{R}+i\rho_{10}^{I} & 1-\rho_{00}
    \end{pmatrix},\qquad
    M = \begin{pmatrix}
    M_{00} & M_{01}^{R}+iM_{01}^{I} \\
    M_{10}^{R}+iM_{10}^{I} & M_{11}
    \end{pmatrix},
\end{equation}
where hermiticity and the normalization $\Tr[\rho]=1$ have been taken into account. This reduces the problem to an optimization over seven distinct real parameters, which makes it computationally tractable. While this approach is feasible for small Hilbert spaces, computational cost grows rapidly with dimension. 

\end{document}